\documentclass[journal]{IEEEtran}
\ifCLASSINFOpdf

\else

\fi
\usepackage{amsmath}
\usepackage{makeidx}  
\usepackage{algorithm}
\usepackage{algorithmic}
\usepackage{graphicx}
\usepackage{subfigure}
\usepackage{epstopdf}
\usepackage{bm}
\usepackage{cite}
\usepackage{stfloats}

\usepackage{amssymb}
\usepackage{graphicx}

\usepackage{url}

\usepackage{tabularx,booktabs}
\newcolumntype{C}{>{\centering\arraybackslash}X} 
\usepackage{lipsum}

\usepackage{makecell} 

\usepackage{graphicx}
\usepackage{color}
\newtheorem{thm}{Theorem}
\newtheorem{rem}{Remark}

\newtheorem{pos}{Proposition}
\newtheorem{proof}{proof}

\begin{document}

\title{Sum Capacity Characterization of Pinching Antennas-assisted Multiple Access Channels}

\author{Guangji Chen,
        Qingqing Wu,
        Kangda Zhi,
        Xidong Mu,
        and Yuanwei Liu, \emph{Fellow, IEEE} \vspace{-22pt}
        \thanks{Guangji Chen is with Nanjing University of Science and Technology, Nanjing 210094, China (email:
                guangjichen@njust.edu.cn). Qingqing Wu is with Shanghai Jiao Tong University, 200240, China
                (e-mail: qingqingwu@sjtu.edu.cn). Kangda Zhi is with Technische Universit at Berlin, 10587 Berlin, Germany (e-mail: k.zhi@tu-berlin.de). Xidong Mu is with Queen's University Belfast, Belfast, BT3 9DT, U.K. (email: x.mu@qub.ac.uk). Yuanwei Liu is with the University of Hong Kong, Hong Kong (e-mail: yuanwei@hku.hk).}}

\maketitle
\vspace{-3pt}
\begin{abstract}
Pinching antenna system (PASS) has recently shown its promising ability to flexibly reconfigure wireless channels via dynamically adjusting the positions of pinching antennas over a dielectric waveguide, termed as \emph{pinching beamforming.} This paper studies the fundamental limit of the sum rate for a PASS-assisted multiple access channel, where multiple users transmit individual messages to a base station under the average power constraint. To this end, a \emph{dynamic pinching beamforming} setup is conceived, where multiple pinching beamforming vectors are employed in a transmission period and the capacity-achieving non-orthogonal multiple access (NOMA) based scheme is considered. For the ideal case with an asymptotically large number of pinching beamforming vectors, the optimal transmission scheme is unveiled to carry out alternating transmission among each user whose channel power gain is maximized with the tailored pinching beamforming. This implies that NOMA is not needed for achieving the sum capacity and the required optimal number of pinching beamforming vectors is equal to the number of users. With this insight, the corresponding sum rate is derived in closed-form expression, which serves as the upper bound of the sum rate. Inspired by this result, a lower bound of the sum rate under an arbitrarily finite number of pinching beamforming vectors is obtained. Numerical results validate our theoretical findings and also illustrate the practical significance of using dynamic pinching beamforming to improve the sum rate.
\end{abstract}


\IEEEpeerreviewmaketitle

\vspace{-10pt}
\section{Introduction}

\vspace{-2pt}
Driven by the recent advancements toward sixth-generation (6G) networks, various reconfigurable-antenna technologies, such as intelligent reflecting surfaces \cite{wu2025intelligent}, movable antennas \cite{zhu2023modeling}, and fluid antennas \cite{wong2020fluid}, have attracted significant research interests for their capability to proactively reconfigure wireless channels. Despite their potential to improve the network performance, these reconfigurable-antenna technologies suffer the limitations on mitigating large-scale path loss due to the confined apertures  with a few wavelengths. To address this limitation, a new type of reconfigurable-antenna technologies, namely pinching antenna system (PASS), has been introduced by NTT DOCOMO \cite{ding2025flexible}. Specifically, a PASS consists of specialized transmission lines, i.e., dielectric waveguide, and simple dielectric materials, called pinching antennas (PAs). By adjusting positions of PAs along a dielectric waveguide, termed as \emph{pinching beamforming}, the short-range line-of-sight (LoS) links between transceivers can be created, which offers a viable solution to mitigate the path loss via flexible antenna deployments \cite{liu2025pinching}.

To fully unleash the potential of PASS for channel reconfigurations, the joint optimization of pinching beamforming and resource allocation has received growing interest to improve the wireless network performance \cite{xu2025rate,tegos2025minimum,ouyang2025array,ren2025pinching,ouyang2025capacity}. Among others, exploring the information-theoretic channel capacity limits of PASS is crucial to understand its maximum achievable gains. To unveil the fundamental limit of PASS under a single-user setup, the authors of \cite{ouyang2025array} investigated the maximum array gain and identified the optimal number of PAs. Nevertheless, for a multi-user PASS, the interplay between the pinching beamforming and multiple access schemes should be carefully captured to achieve the performance upper bound. For this consideration, the authors of \cite{ren2025pinching} studied the power minimization problem in a downlink two-user PASS by considering both non-orthogonal multiple access (NOMA) and orthogonal multiple access (OMA) schemes. From the information-theoretic viewpoint, the capacity region of a two-user PASS-aided wireless communication system was characterized based on the capacity-achieving NOMA scheme \cite{ouyang2025capacity}. Although these works \cite{ouyang2025array,ren2025pinching,ouyang2025capacity} provided concrete insights on the performance limits of PASS, the results are limited to a two-user case, which may not be directly applicable to a general multi-user setup. For example, considering a fundamental PASS-assisted multiple access channel (MAC) under the average power constraint \cite{shang2008sum}, the optimal transmission scheme for achieving the sum capacity is still unknown.

\begin{figure}[!t]
\centering
\includegraphics[width= 0.4\textwidth]{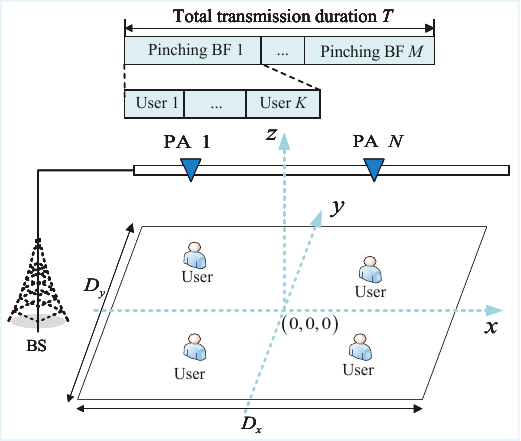}
\DeclareGraphicsExtensions.
\caption{A PASS-assisted MAC with dynamic pinching beamforming}
\label{model}
\vspace{-16pt}
\end{figure}

To fill in this gap, we study a PASS-assisted MAC, where multiple users transmit individual messages to a base station (BS) under the average power constraint. For revealing the sum rate limit, we consider a \emph{dynamic pinching beamforming} setup, where the positions of PAs can be reconfigured multiple times in a transmission period, which enables the flexibility in enhancing the multi-user diversity. In each time slot (TS) with a fixed pinching beamforming vector, the capacity-achieving NOMA scheme is employed. Under this setup, we jointly optimize the pinching beamforming and resource allocation over time to unveil the fundamental sum rate limit. We first focus on the ideal case with an asymptotically large number of available pinching beamforming vectors. In this case, we theoretically prove that the optimal capacity-achieving transmission scheme is to carry out alternating transmission among each user whose channel power gain is maximized with the tailored pinching beamforming. The result sheds light on the sum-capacity optimality of OMA in PASS-assisted  MAC, which indicates that NOMA, i.e., successive interference cancellation (SIC)-based receiver, is not required, thereby resulting in a simplified transceiver structure. Then, a lower bound of the sum rate under an arbitrarily finite number of pinching beamforming vectors is constructed. Numerical results are provided to verify our theoretical findings and to unveil the performance-cost tradeoff in exploiting dynamic pinching beamforming.

\vspace{-8pt}
\section{System Model }
As illustrated in Fig. \ref{model}, we consider a fundamental PASS-assisted MAC, where $K$ users send their individual messages to a single-waveguide PASS-based BS over a finite transmission period, which is denoted by ${\cal T} \buildrel \Delta \over = \left( {0,T} \right]$ with duration $T > 0$. In a three-dimensional (3D) coordination system, the $K$ users indexed by $k \in {\cal K} \buildrel \Delta \over = \left\{ {1, \ldots ,K} \right\}$ are randomly distributed in a
rectangular region centered at the origin with side lengths of ${D_x}$ and ${D_y}$. Let ${{\bf{w}}_k} = \left[ {{x_{u,k}},{y_{u,k}},0} \right]$ denote the location of user $k$. The locations of $K$ users remain unchanged for $\forall t \in {\cal T}$. Without loss of generality, we assume that the waveguide is deployed parallel to the $x$-axis with a height of $d$. The waveguide is equipped with $N$ PAs to enhance the uplink transmission via implementing pinching beamforming. The set of the index for all PAs is denoted by ${\cal N} \buildrel \Delta \over = \left\{ {1, \ldots ,N} \right\}$. Let ${{\bf{v}}_0} = [ - {D_x}/2,0,d]$ denote the location of the feed point at the waveguide.

To fully harness the potential gain introduced by PASS and characterize the fundamental limit of the sum rate, we consider a setup of \emph{dynamic pinching beamforming}, where the positions of PAs can be reconfigured multiple times in the transmission period ${\cal T}$. In particular, we divide the whole transmission period into $M$ TSs and each of them is denoted by ${{\cal T}_m} \buildrel \Delta \over = \left( {\sum\nolimits_{i = 1}^{m - 1} {{t_i},\sum\nolimits_{i = 1}^m {{t_i}} } } \right]$, $\forall m \in {\cal M} \buildrel \Delta \over = \left\{ {1, \ldots ,M} \right\}$. The duration of the $m$-th TS is denoted by ${{t_m}}$, which satisfies $\sum\nolimits_{m = 1}^M {{t_m} = T}$. In the $m$-th TS, the position of PA $n$ is adjusted to ${{\bf{v}}_n}\left[ m \right] = \left[ {{v_n}\left[ m \right],0,d} \right], \forall n \in {\cal N}$. The vector of $x$-axes for all PAs is denoted by ${\bf{\bar v}}\left[ m \right] = {\left[ {{v_1}\left[ m \right], \ldots ,{v_N}\left[ m \right]} \right]^{\mathop{T}\nolimits} } \in {\mathbb{R}^{N \times 1}}$, which is refered as the \emph{pinching beamforming pattern} employed in the $m$-th TS. We assume that the PAs can flexibly adjust their positions continuously along the waveguide with a length of ${L_{\max }}$. Let $\Delta$ denote the minimum inter-antenna spacing to avoid the EM mutual coupling. Without loss of generality, the feasible set of ${\bf{\bar v}}\left[ m \right]$ is given by
\begin{align}\label{feasible_set}
{\cal V} \buildrel \Delta \over = \left\{ {{v_n}\left| {0 \le {v_n} \le {L_{\max }},{v_{n + 1}} - {v_n} \ge \Delta ,\forall n} \right.} \right\}.
\end{align}
To facilitate the characterization of the maximum achievable performance, we consider the free-space LoS channel model, which is valid due to the ability of PASS to create the LoS-dominant channel. Under the LoS channel assumption, the wireless channel coefficient from user $k$ to the $n$-th PA in the $m$-th TS is given by
\begin{align}\label{channel_cofficient}
{h_k}\left( {{{\bf{v}}_n}\left[ m \right]} \right) = \frac{{{\eta ^{1/2}}{e^{ - j\frac{{2\pi }}{\lambda }\left\| {{{\bf{v}}_n}\left[ m \right] - {{\bf{w}}_k}} \right\|}}}}{{\left\| {{{\bf{v}}_n}\left[ m \right] - {{\bf{w}}_k}} \right\|}},
\end{align}
where $\lambda $ denotes the wavelength and $\eta  = \frac{{{\lambda ^2}}}{{16{\pi ^2}}}$. Hence, the wireless channel vector from user $k$ to all PAs in the $m$-th TS is ${{\bf{h}}_k}\left( {{\bf{\bar v}}\left[ m \right]} \right) = {\left[ {{h_k}\left( {{{\bf{v}}_1}\left[ m \right]} \right), \ldots ,{h_k}\left( {{{\bf{v}}_N}\left[ m \right]} \right)} \right]^T}$.

Then, we introduce the uplink transmission process in each TS. Let ${s_k}\left[ m \right] \sim {\cal C}{\cal N}\left( {0,1} \right)$ denote the information symbol of user $k$ in the $m$-th TS, $\forall m \in {\cal M}$. The received signal at PA $n$ in the $m$-th TS can be expressed as
\begin{align}\label{received_signal_per_antenna}
{y_n}\left[ m \right] \!\!=\!\! \sum\nolimits_{k = 1}^K {{h_k}\left( {{{\bf{v}}_n}\left[ m \right]} \right){p_k}\left[ m \right]} {s_k}\left[ m \right] \!\!+\!\! {z_n}\left[ m \right],
\end{align}
where ${{p_k}\left[ m \right]}$ and ${z_n}\left[ m \right]\sim {\cal C}{\cal N}\left( {0,{\sigma ^2}} \right)$ denote the transmit power of user $k$ and the additive Gaussian white noise (AWGN) with power ${\sigma ^2}$ in the $m$-th TS, respectively. Let ${\bf{g}}\left( {{\bf{\bar v}}\left[ m \right]} \right) = {\left[ {{e^{ - j\frac{{2\pi }}{{{\lambda _{\rm{g}}}}}\left\| {{{\bf{v}}_1}\left[ m \right] - {{\bf{v}}_0}} \right\|}}, \ldots ,{e^{ - j\frac{{2\pi }}{{{\lambda _{\rm{g}}}}}\left\| {{{\bf{v}}_N}\left[ m \right] - {{\bf{v}}_0}} \right\|}}} \right]^T}$ denote the in-waveguide channel vector from all PAs to the feed point, where ${\lambda _g} = \lambda /{n_{{\rm{eff}}}}$ is the guided wavelength with ${n_{{\rm{eff}}}}$ denoting the effective refractive index of the waveguide \cite{ding2025flexible}. Accordingly, the aggregated signal received at the feed point is given by
\begin{align}\label{received_signal_feed_point}
y\left[ m \right] &= \sum\nolimits_{n = 1}^N {{e^{ - j\frac{{2\pi }}{{{\lambda _{\rm{g}}}}}\left\| {{{\bf{v}}_n}\left[ m \right] - {{\bf{v}}_0}} \right\|}}{y_n}\left[ m \right]}\nonumber\\
& = \sum\limits_{k = 1}^K {{{\bf{g}}^T}\left( {{\bf{\bar v}}\left[ m \right]} \right)} {{\bf{h}}_k}\left( {{\bf{\bar v}}\left[ m \right]} \right){p_k}\left[ m \right]{s_k}\left[ m \right] \!+\! z\left[ m \right],
\end{align}
where $z\left[ m \right] = {\sum\nolimits_{n = 1}^N {\left[ {{\bf{g}}\left( {{\bf{\bar v}}\left[ m \right]} \right)} \right]} _n}{z_n}\left[ m \right]$ satisfying $z\left[ m \right] \sim {\cal C}{\cal N}\left( {0,N{\sigma ^2}} \right)$. Under the given transmit power ${p_k}\left[ m \right]$, it is well-known that the capacity-achieving scheme is NOMA transmission with SIC at the BS receiver \cite{shang2008sum}. By employing NOMA, the achievable sum rate in the $m$-th TS is given by
\begin{align}\label{sumrate_perTS}
R\left[ m \right] \!\!=\!\! {\log _2}\left( {1 \!\!+\!\! \frac{{\sum\nolimits_{k = 1}^K {{p_k}\left[ m \right]{{\left| {{{\bf{g}}^T}\left( {{\bf{\bar v}}\left[ m \right]} \right){{\bf{h}}_k}\left( {{\bf{\bar v}}\left[ m \right]} \right)} \right|}^2}} }}{{N{\sigma ^2}}}} \right).
\end{align}

Under the dynamic pinching beamforming configuration $\left\{ {{\bf{\bar v}}\left[ m \right],\forall m \in {\cal M}} \right\}$ over all TSs, the average sum rate in the considered transmission period ${\cal T}$ can be written as
\begin{align}\label{average_rate}
\bar R = \frac{1}{T}\sum\nolimits_{m = 1}^M {{t_m}R\left[ m \right]}.
\end{align}
Suppose that each user has a maximum average transmit power budget ${P_{{\rm{ave}}}}$. Thus, we have $\left( {\sum\nolimits_{m = 1}^M {{p_k}\left[ m \right]} {t_m}} \right)/T \le {P_{{\rm{ave}}}}, \forall k$. We are interested in characterizing the sum capacity of the considered PASS-assisted MAC, which corresponds to maximize the average sum rate defined in \eqref{average_rate} via jointly optimizing the dynamic pinching beamforming $\left\{ {{\bf{\bar v}}\left[ m \right],\forall m \in {\cal M}} \right\}$, the transmit power $\left\{ {{p_k}\left[ m \right],\forall k,m} \right\}$ of users, and the time allocation $\left\{ {{t_m},\forall m} \right\}$, subject to the average power constraints and continuous antenna position constraints. As a result, the corresponding optimization problem can be formulated as follows:
\begin{subequations}\label{C1}
\begin{align}
\label{C1-a}\mathop {\max }\limits_{\left\{ {{\bf{\bar v}}\left[ m \right],{p_k}\left[ m \right],{t_m}} \right\}}\;\;&\frac{1}{T}\sum\nolimits_{m = 1}^M {{t_m}R\left[ m \right]}\\
\label{C1-b}{\rm{s.t.}}\;\;\;\;\;\;\;\;\;\;&\frac{1}{T}\sum\nolimits_{m = 1}^M {{p_k}\left[ m \right]} {t_m} \le {P_{{\rm{ave}}}}, ~\forall k \in {\cal K},\\
\label{C1-c}&{{p_k}\left[ m \right]} \ge 0,~\forall k \in {\cal K}, {\forall m \in {\cal M}},\\
\label{C1-d}&{v_n}\left[ m \right] \in {\cal V}, ~\forall n \in {\cal N}, {\forall m \in {\cal M}},\\
\label{C1-d}&\sum\nolimits_{m = 1}^M {{t_m} = T}.
\end{align}
\end{subequations}
It is observed that problem \eqref{C1} cannot be decomposed into $M$ independent subproblems directly since the transmit power is coupled over $M$ TSs. Without loss of optimality, all the variables should be optimized jointly to maximize the average sum rate. Besides, problem \eqref{C1} is non-convex and challenging to solve optimally since the optimization variables are highly coupled in \eqref{C1-a} and \eqref{C1-b}.

\section{Sum Capacity Characterization}
In this section, we characterize the sum capacity the PASS-assisted MAC by solving problem \eqref{C1}. We first obtain the fundamental limit of the sum rate by solving problem \eqref{C1} optimally under the case of $M \to \infty$. Then, we further derive a lower bound of the sum rate with an arbitrarily finite $M$.
\subsection{Sum Rate with $M \to \infty$}
In this subsection, we first investigate an ideal case of $M \to \infty$, which serves as an upper bound of the achievable sum rate. In this case, we replace the discrete index $\left[ m \right]$ by the continuous index $\left( t \right)$. Then, problem \eqref{C1} can be equivalently expressed as the following problems with continuous time variables.
\begin{subequations}\label{C2}
\begin{align}
\label{C2-a}\mathop {\max }\limits_{\left\{ {{\bf{\bar v}}\left( t \right),{p_k}\left( t \right)} \right\}}\;\;&\frac{{\int_0^T {{{\log }_2}\left( {1 + \frac{{\sum\nolimits_{k = 1}^K {{p_k}\left( t \right){{\left| {{\gamma _k}\left( {{\bf{\bar v}}\left( t \right)} \right)} \right|}^2}} }}{{N{\sigma ^2}}}} \right)} dt}}{T}\\
\label{C2-b}{\rm{s.t.}}\;\;\;\;\;\;&\frac{1}{T}\int_0^T {{p_k}\left( t \right)dt \le {P_{{\rm{ave}}}}}, ~\forall k \in {\cal K},\\
\label{C2-c}&{{p_k}\left( t \right)} \ge 0,~\forall k \in {\cal K}, t \in {\cal T},\\
\label{C2-d}&{v_n}\left( t \right) \in {\cal V}, ~\forall n \in {\cal N}, t \in {\cal T},
\end{align}
\end{subequations}
where ${\gamma _k}\left( {{\bf{\bar v}}\left( t \right)} \right) = {{\bf{g}}^T}\left( {{\bf{\bar v}}\left( t \right)} \right){{\bf{h}}_k}\left( {{\bf{\bar v}}\left[ t \right]} \right)$. Problem \eqref{C2} is highly non-convex, which contains infinite optimization variables over the continuous interval ${\cal T}$. To simplify the optimization of \eqref{C2}, we have the following proposition.
\begin{pos}
Problem \eqref{C2} satisfies the time-sharing condition defined in \cite{yu2006dual}.
\end{pos}
\begin{proof}
Let $\left\{ {p_{k,x}^*\left( t \right),{\bf{\bar v}}_{x}^*\left( t \right)} \right\}$ and $\left\{ {p_{k,y}^*\left( t \right),{\bf{\bar v}}_{y}^*\left( t \right)} \right\}$ denote the optimal solutions to problem \eqref{C2} with the average power budgets ${P_{{\rm{ave,}}x}}$ and ${P_{{\rm{ave,}}y}}$, respectively. Their corresponding objective values are denoted by $\bar R_x^*$ and $\bar R_y^*$, respectively. Thus, $\left\{ {p_{k,x}^*\left( t \right),{\bf{\bar v}}_{x}^*\left( t \right)} \right\}$ and $\left\{ {p_{k,y}^*\left( t \right),{\bf{\bar v}}_{y}^*\left( t \right)} \right\}$ satisfy constraints \eqref{C2-d} and $\frac{1}{T}\int_0^T {p_{k,i}^*\left( t \right)dt \le {P_{{\rm{ave,}}i}}} ,i \in \left\{ {x,y} \right\}$. To demonstrate that problem \eqref{C2} satisfies the time-sharing condition, we need to show that there exists a feasible solution $\left\{ {{p_{k,z}}\left( t \right),{{{\bf{\bar v}}}_z}\left( t \right)} \right\}$ such that it (i) satisfies the average power constraint with $v{P_{{\rm{ave,}}x}} + \left( {1 - v} \right){P_{{\rm{ave,}}y}}$; and (ii) achieves an equal or a higher average sum rate than $v\bar R_x^* + \left( {1 - v} \right)\bar R_y^*$ under any $0 \le v \le 1$. To this end, we construct $\left\{ {{p_{k,z}}\left( t \right),{{{\bf{\bar v}}}_z}\left( t \right)} \right\}$ as follows:
\begin{align}\label{constructed_solution1}
{p_{k,z}}\left( t \right) = \left\{ {\begin{array}{*{20}{c}}
{p_{k,x}^*\left( {t/v} \right)~~~~~~~~~~~~~~~~~~{\rm{                    0}} \le t \le vT}\\
{p_{k,y}^*\left( {\left( {t - vT} \right)/\left( {1 - v} \right)} \right){\rm{ }}~~vT \le t \le T},
\end{array}} \right.
\end{align}
\begin{align}\label{constructed_solution2}
{{{\bf{\bar v}}}_z}\left( t \right) = \left\{ {\begin{array}{*{20}{c}}
{{\bf{\bar v}}_x^*\left( {t/v} \right)~~~~~~~~~~~~~~~~~~~~~~{\rm{                    0}} \le t \le vT}\\
{{\bf{\bar v}}_y^*\left( {\left( {t - vT} \right)/\left( {1 - v} \right)} \right){\rm{ }}~~~~vT \le t \le T}.
\end{array}} \right.
\end{align}
Denote $R_i^*\left( t \right)$ as the instantaneous rate under the solution of $\left\{ {p_{k,i}^*\left( t \right),{\bf{\bar v}}_{i}^*\left( t \right)} \right\}$, $i \in \left\{ {x,y} \right\}$. Based on \eqref{constructed_solution1} and \eqref{constructed_solution2}, the objective value ${{\bar R}_z}$ achieved by $\left\{ {{p_{k,z}}\left( t \right),{{{\bf{\bar v}}}_z}\left( t \right)} \right\}$ can be obtained as follows:
\begin{align}\label{rate_z}
{{\bar R}_z} &= \frac{1}{T}\int_0^{vT} {R_x^*\left( {\frac{t}{v}} \right)dt}  + \frac{1}{T}\int_{vT}^T {R_x^*\left( {\frac{{t - vT}}{{1 - v}}} \right)dt}\nonumber\\
& = \frac{v}{T}\int_0^T {R_x^*\left( \omega  \right)d\omega  + \frac{{1 - v}}{T}} \int_0^T {R_y^*\left( \omega  \right)d\omega }\nonumber\\
& = v\bar R_x^* + \left( {1 - v} \right)\bar R_y^*,
\end{align}
which indicates that the objective value achieved by $\left\{ {{p_{k,z}}\left( t \right),{{{\bf{\bar v}}}_z}\left( t \right)} \right\}$ equals to $v\bar R_x^* + \left( {1 - v} \right)\bar R_y^*$. Similarly, it can be shown that $\frac{1}{T}\int_0^T {{p_{k,z}}\left( t \right)dt \le v{P_{{\rm{ave,}}x}} + \left( {1 - v} \right){P_{{\rm{ave,}}y}}}$, which thus completes the proof.
\end{proof}

\textbf{Proposition 1} indicates that problem \eqref{C2} meets the time-sharing condition \cite{yu2006dual}, which leads to the strong duality holding for problem \eqref{C2}. Therefore, the optimal solution to problem \eqref{C2} can be obtained via its dual problem. By employing the Lagrange duality method, the Lagrange dual function of problem \eqref{C2} is given by
\begin{align}\label{dual function}
{f_{\rm{1}}}\left( {\left\{ {{\lambda _k}} \right\}} \right) = &\mathop {\max }\limits_{\left\{ {{\bf{\bar v}}\left( t \right),{p_k}\left( t \right)} \right\}} {{\cal L}_1}\left( {{\bf{\bar v}}\left( t \right),\left\{ {{p_k}\left( t \right)} \right\},\left\{ {{\lambda _k}} \right\}} \right)\nonumber\\
&{\rm{s.t.}}~\eqref{C2-c},\eqref{C2-d},
\end{align}
where
\begin{align}\label{Lagrange_function}
\begin{array}{*{20}{l}}
{{\cal L}_1}\left( {{\bf{\bar v}}\left( t \right),\left\{ {{p_k}\left( t \right)} \right\},\left\{ {{\lambda _k}} \right\}} \right)\\
{ = \frac{1}{T}\int_0^T {{{\log }_2}\left( {1 + \frac{{\sum\nolimits_{k = 1}^K {{p_k}\left( t \right){{\left| {{\gamma _k}\left( {{\bf{\bar v}}\left( t \right)} \right)} \right|}^2}} }}{{N{\sigma ^2}}}} \right)} dt}\\
 ~~ - \sum\nolimits_{k = 1}^K {{\lambda _k}\left( {\frac{1}{T}\int_0^T {{p_k}\left( t \right)} dt - {P_{{\rm{ave}}}}} \right)}
\end{array}
\end{align}
is the Lagrange function with ${\left\{ {{\lambda _k}} \right\}}$ denoting the non-negative Lagrange multipliers associated with \eqref{C2-b}. Accordingly, the dual problem of problem \eqref{C2} is given by
\begin{align}\label{dual problem}
\mathop {\min }\limits_{\left\{ {{\lambda _k}} \right\}} {f_{\rm{1}}}\left( {\left\{ {{\lambda _k}} \right\}} \right), ~~{\rm{s.t.}}~{\lambda _k} \ge 0.
\end{align}
Based on the strong duality, problem \eqref{C2} can be optimally solved by solving its dual problem. With the optimal dual variables, denoted by $\left\{ {\lambda _k^\star} \right\}$, at hand, suppose that there are $\Pi  \ge 1$ optimal solutions to problem \eqref{dual function}, denoted by $\left\{ {{\bf{\bar v}}_\varpi ^\star\left( t \right),p_{k,\varpi }^\star\left( t \right)} \right\}_{\varpi  = 1}^\Pi$. The optimal primal solution of problem \eqref{C2} needs to be reconstructed via time sharing among all these solutions, i.e., allocating a particular duration for each solution in $\left\{ {{\bf{\bar v}}_\varpi ^\star\left( t \right),p_{k,\varpi }^\star\left( t \right)} \right\}_{\varpi  = 1}^\Pi$. In the following, we derive the optimal solution to problem \eqref{C2} in closed-form expressions by exploiting its inherent structures.

Define ${\bf{\bar v}}_k^*\left( t \right) = \mathop {\arg \max }\limits_{{\bf{\bar v}}\left( t \right)} {\left| {{\gamma _k}\left( {{\bf{\bar v}}\left( t \right)} \right)} \right|^2}$. It can be found that the optimal value of problem \eqref{C2} is upper bounded by
\begin{align}\label{rate_upperbound}
\bar R = {\log _2}\left( {1 + \frac{{\sum\nolimits_{k = 1}^K {{P_{{\rm{ave}}}}{{\left| {{\gamma _k}\left( {{\bf{\bar v}}_k^*\left( t \right)} \right)} \right|}^2}} }}{{N{\sigma ^2}}}} \right).
\end{align}
This is because that $g\left( {x,y} \right) = x{\log _2}\left( {1 + y/x} \right)$ is a concave function with respect to $x$ and $y$. Then, we construct a feasible solution of problem \eqref{C2} to show that the bound $\bar R$ in \eqref{rate_upperbound} is achievable. To this end, we have the following proposition.
\begin{pos}
Under the given dual variables
\begin{align}\label{dual_varialbes_optimal}
\lambda _k^* = \frac{{{{\left| {{\gamma _k}\left( {{\bf{\bar v}}_k^*\left( t \right)} \right)} \right|}^2}}}{{N{\sigma ^2} + \sum\nolimits_{k = 1}^K {{P_{{\rm{ave}}}}{{\left| {{\gamma _k}\left( {{\bf{\bar v}}_k^*\left( t \right)} \right)} \right|}^2}} }},\forall k,
\end{align}
there are $K$ solutions of \eqref{dual function}, denoted by $\left\{ {{\Gamma _k},\forall k \in {\cal K}} \right\}$, can achieve the value of $\bar R$ in \eqref{rate_upperbound}. In particular, the $K$ solutions are given by ${\Gamma _k} = \left\{ {{\bf{\bar v}}\left( t \right),\left\{ {{p_k}\left( t \right)} \right\}} \right\} = \left\{ {{\bf{\bar v}}_k^*\left( t \right),\left\{ {{{\bf{0}}_{k - 1}},p_k^*\left( t \right),{{\bf{0}}_{K - k}}} \right\}} \right\}$, where
\begin{align}\label{optimal_solution1}
p_k^*\left( t \right) = \frac{{\sum\nolimits_{i = 1}^K {{P_{{\rm{ave}}}}{{\left| {{\gamma _i}\left( {{\bf{\bar v}}_i^*\left( t \right)} \right)} \right|}^2}} }}{{{{\left| {{\gamma _k}\left( {{\bf{\bar v}}_k^*\left( t \right)} \right)} \right|}^2}}},\forall k.
\end{align}
\end{pos}
\begin{proof}
We first focus on an arbitrary $k$, $k \in {\cal K}$. By setting ${{{\bf{\bar v}}}}\left( t \right) = {\bf{\bar v}}_k^*\left( t \right)$, we consider the optimization problem as
\begin{subequations}\label{C3}
\begin{align}
\label{C3-a}\mathop {\max }\limits_{\left\{ {{p_i}\left( t \right)} \right\}}\;\;&{{\cal L}_1}\left( {{\bf{\bar v}}_k^*\left( t \right),\left\{ {{p_i}\left( t \right)} \right\},\left\{ {\lambda _i^*} \right\}} \right)\\
\label{C3-b}{\rm{s.t.}}\;\;\;\;&{{p_i}\left( t \right)} \ge 0,~\forall i \in {\cal K}.
\end{align}
\end{subequations}
For problem \eqref{C3}, it is equivalent to
\begin{align}\label{power_optimization}
\mathop {\max }\limits_{\left\{ {{p_i}\left( t \right) \ge 0} \right\}} \; &{\log _2}\left( {1 \!+\! \frac{{\sum\nolimits_{i = 1}^K {{p_i}\left( t \right){{\left| {{\gamma _i}\left( {{\bf{\bar v}}_k^*\left( t \right)} \right)} \right|}^2}} }}{{N{\sigma ^2}}}} \right) \!-\!
 \sum\limits_{i = 1}^K {\lambda _i^*{p_i}\left( t \right)}.
\end{align}
Denote the objective function in \eqref{power_optimization} as ${f_2}\left( {\left\{ {{p_i}\left( t \right)} \right\}} \right)$. It is evident that ${f_2}\left( {\left\{ {{p_i}\left( t \right)} \right\}} \right)$ is concave with respect to ${\left\{ {{p_k}\left( t \right)} \right\}}$. By taking the derivative of $ {f_2}\left( {\left\{ {p_i^*\left( t \right)} \right\}} \right)$ with respect to ${p_i^*\left( t \right)}$, we have
\begin{align}\label{derivation}
\frac{{\partial {f_2}\left( {\left\{ {p_i^*\left( t \right)} \right\}} \right)}}{{\partial p_i^*\left( t \right)}} = \frac{{{{\left| {{\gamma _i}\left( {{\bf{\bar v}}_k^*\left( t \right)} \right)} \right|}^2}}}{{1 + \sum\nolimits_{i = 1}^K {p_i^*\left( t \right){{\left| {{\gamma _k}\left( {{\bf{\bar v}}_k^*\left( t \right)} \right)} \right|}^2}} }} - \lambda _i^*.
\end{align}
Since ${{\lambda _i^*} \mathord{\left/
 {\vphantom {{\lambda _i^*} {\lambda _k^* = }}} \right.
 \kern-\nulldelimiterspace} {\lambda _k^* = }}{{{{\left| {{\gamma _i}\left( {{\bf{\bar v}}_i^*\left( t \right)} \right)} \right|}^2}} \mathord{\left/
 {\vphantom {{{{\left| {{\gamma _i}\left( {{\bf{\bar v}}_i^*\left( t \right)} \right)} \right|}^2}} {{{\left| {{\gamma _k}\left( {{\bf{\bar v}}_k^*\left( t \right)} \right)} \right|}^2}}}} \right.
 \kern-\nulldelimiterspace} {{{\left| {{\gamma _k}\left( {{\bf{\bar v}}_k^*\left( t \right)} \right)} \right|}^2}}}$ and ${\left| {{\gamma _i}\left( {{\bf{\bar v}}_k^*\left( t \right)} \right)} \right|^2} \le {\left| {{\gamma _i}\left( {{\bf{\bar v}}_i^*\left( t \right)} \right)} \right|^2}$, $\forall i \ne k$, we obtain
 \begin{align}\label{derivation1}
\frac{{\partial {f_2}\left( {\left\{ {p_i^*\left( t \right)} \right\}} \right)}}{{\partial p_k^*\left( t \right)}} = 0,\frac{{\partial {f_2}\left( {\left\{ {p_i^*\left( t \right)} \right\}} \right)}}{{\partial p_i^*\left( t \right)}} \le 0,\forall i \ne k,
\end{align}
which directly leads to
\begin{align}\label{optimal_power}
p_k^*\left( t \right) = \frac{{{P_{{\rm{ave}}}}\sum\nolimits_{i = 1}^K {{{\left| {{\gamma _i}\left( {{\bf{\bar v}}_i^*\left( t \right)} \right)} \right|}^2}} }}{{{{\left| {{\gamma _k}\left( {{\bf{\bar v}}_k^*\left( t \right)} \right)} \right|}^2}}},p_i^*\left( t \right) = 0,\forall i \ne k.
\end{align}
Then, it is not difficult to show that
\begin{align}\label{temp}
{{\cal L}_1}\left( {{\Gamma _k},\left\{ {\lambda _i^*} \right\}} \right) \!\!=\!\! {\log _2}\left( {1 \!\!+\!\! \frac{{{P_{{\rm{ave}}}}\sum\nolimits_{k = 1}^K {{{\left| {{\gamma _k}\left( {{\bf{\bar v}}_k^*\left( t \right)} \right)} \right|}^2}} }}{{N{\sigma ^2}}}} \right),\forall k.
\end{align}
Thus, we complete the proof.
\end{proof}

Based on \textbf{Proposition 2}, we further show that ${\bar R}$ is achievable for problem \eqref{C2} by constructing its feasible solution via $\left\{ {{\Gamma _k},\forall k \in {\cal K}} \right\}$. By performing time-sharing among all $K$ solutions in $\left\{ {{\Gamma _k},\forall k \in {\cal K}} \right\}$ and letting ${\tau _k}$ denote the duration for solution ${{\Gamma _k}}$, we have the following problem:
\begin{subequations}\label{C4}
\begin{align}
\label{C4-a}\mathop {\max }\limits_{{{\tau _k}}}\;\;&\sum\nolimits_{k = 1}^K {{\tau _k}} {\log _2}\left( {1 + \frac{{p_k^*\left( t \right){{\left| {{\gamma _k}\left( {{\bf{\bar v}}_k^*\left( t \right)} \right)} \right|}^2}}}{{N{\sigma ^2}}}} \right)\\
\label{C4-b}{\rm{s.t.}}\;\;\;&\sum\nolimits_{k = 1}^K {{\tau _k}}  = T,\\
\label{C4-c}&\frac{1}{T}{\tau _k}p_k^*\left( t \right) \le {P_{{\rm{ave}}}},~\forall k \in {\cal K}.
\end{align}
\end{subequations}
It is evident that the optimal solution of problem \eqref{C4}, denoted by $\left\{ {\tau _k^*} \right\}$, is given by
\begin{align}\label{optimal_time_allocation}
\tau _k^* = T\frac{{{{\left| {{\gamma _k}\left( {{\bf{\bar v}}_k^*\left( t \right)} \right)} \right|}^2}}}{{\sum\nolimits_{i = 1}^K {{{\left| {{\gamma _i}\left( {{\bf{\bar v}}_i^*\left( t \right)} \right)} \right|}^2}} }},\forall k.
\end{align}
Finally, with $\left\{ {\tau _k^*} \right\}$ and $\left\{ {{\Gamma _k},\forall k \in {\cal K}} \right\}$ at hand, we divide the whole transmission period ${\cal T}$ into $K$ portions, denoted by periods ${{\cal T}_1}, \ldots ,{{\cal T}_K}$, where ${{\cal T}_k} = \left( {\sum\nolimits_{i = 1}^{k - 1} {\tau _i^*,\sum\nolimits_{i = 1}^k {\tau _i^*} } } \right]$. Then, we have the optimal solution of problem \eqref{C2}  in the following theorem.
\begin{thm}
The optimal solution of problem \eqref{C2} is ${{{\bf{\bar v}}}^*}\left( t \right) \!=\! {\bf{\bar v}}_k^*\left( t \right)$, $p_k^*\left( t \right) \!\!=\! \! \frac{{{P_{{\rm{ave}}}}\sum\nolimits_{i = 1}^K {{{\left| {{\gamma _i}\left( {{\bf{\bar v}}_i^*\left( t \right)} \right)} \right|}^2}} }}{{{{\left| {{\gamma _k}\left( {{\bf{\bar v}}_k^*\left( t \right)} \right)} \right|}^2}}},p_i^*\left( t \right) \!= \!0, i \ne k$, $\forall t \in {{\cal T}_k},k \in {\cal K}$ and the correspondingly achievable sum rate is ${\bar R}$.
\end{thm}

\begin{rem}
Under the average power constraint, \textbf{Theorem 1} implies that the sum capacity-achieving scheme is to carry out alternating transmission among each user whose channel power gain is maximized with the tailored pinching beamforming. Thus, the SIC based receiver is not required, which greatly simplifies the transceiver design. Moreover, \textbf{Theorem 1} also unveils that exploiting dynamic pinching beamforming under $M > K$ will provide no sum rate improvement. In other words, at most $K$ pinching beamforming patterns are sufficient to achieve the sum rate limit of the PASS-assisted  MAC. It is worth noting that the analytical results in this subsection are applicable to the case with an arbitrary $M$ satisfying $M \ge K$.
\end{rem}

For ${\bf{\bar v}}_k^*\left( t \right)$, it can be obtained by employing phase refinement based on the initial solution
\begin{align}\label{position_solution}
v_{n,k}^*\left( t \right) = {x_{u,k}} - \left( {\frac{{N - 1}}{2} - n + 1} \right)\Delta ,\forall n,k.
\end{align}
The details can be referred to \cite{xu2025rate}, which are limited due to page limitations. Under the special case of $N = 1$, the sum capacity is given by a closed-form expression as
\begin{align}\label{signle_antenna_rate}
{\bar R} = {\log _2}\left( {1 + \frac{{{P_{{\rm{ave}}}}}}{{{\sigma ^2}}}\sum\nolimits_{k = 1}^K {{{\left( {y_{u,k}^2 + {d^2}} \right)}^{ - 1}}} } \right).
\end{align}
\vspace{-8pt}
\subsection{Sum Rate with Finite $M$}
As mentioned in \textbf{Remark 1}, we only need to focus on the case of $M < K$ in this subsection. Inspired by the optimal solution structure described in \textbf{Theorem 1}, we reconstruct a high-quality solution of problem \eqref{C1} with $M < K$. Specifically, we first sort the maximum channel power gains of all users, i.e. $\left\{ {{\gamma _k}\left( {{\bf{\bar v}}_k^*\left( t \right)} \right),k \in {\cal K}} \right\}$,  in descending order, where $\Psi \left( k \right)$ denotes the order of user $k$. Then, each of the first $M - 1$ users, i.e., $\Psi \left( k \right) = 1, \ldots M - 1$, is assigned a dedicated pinching beamforming vector. All the rest $K - M + 1$ users share a common pinching beamforming vector ${{{{\bf{\bar v}}}_c}}$ to assist uplink transmission. Then, problem \eqref{C1} is transformed to
\begin{subequations}\label{C6}
\begin{align}
\label{C6-a}\mathop {\max }\limits_{{{{\bf{\bar v}}}_c},\left\{ {{t_m}} \right\}}&{t_M}{\log _2}\left( {1 + \frac{{{P_{{\rm{ave}}}}T\sum\nolimits_{\Psi \left( k \right) = M}^K {{{\left| {{\gamma _k}\left( {{{{\bf{\bar v}}}_c}} \right)} \right|}^2}} }}{{{t_M}N{\sigma ^2}}}} \right)\nonumber\\
&+ \sum\nolimits_{\Psi \left( k \right) = 1}^{M - 1} {{t_{\Psi \left( k \right)}}{{\log }_2}\left( {1 \!\!+\!\! \frac{{{P_{{\rm{ave}}}}T{{\left| {{\gamma _k}\left( {{\bf{\bar v}}_k^*} \right)} \right|}^2}}}{{{t_{\Psi \left( k \right)}}N{\sigma ^2}}}} \right)} \\
\label{C6-b}{\rm{s.t.}}\;\;\;&{\left[ {{{{\bf{\bar v}}}_c}} \right]_n} \in {\cal V}, ~\forall n \in {\cal N},\\
\label{C6-d}&\sum\nolimits_{\Psi \left( k \right) = 1}^{M - 1} {{t_{\Psi \left( k \right)}}}  + {t_M} \le T.
\end{align}
\end{subequations}
For problem \eqref{C6}, it is observed that the optimization of ${{{{\bf{\bar v}}}_c}}$ is equivalent to
\begin{align}\label{common_pinching_beamforming}
\mathop {\max }\limits_{{{{\bf{\bar v}}}_c}} \sum\nolimits_{\Psi \left( k \right) = M}^K {{{\left| {{\gamma _k}\left( {{{{\bf{\bar v}}}_c}} \right)} \right|}^2}} ~~~{\rm{s.t.}}~\eqref{C6-b}.
\end{align}
Although problem \eqref{common_pinching_beamforming} is non-convex, we invoke the element-wise alternating optimization-based approach to obtain its high-quality suboptimal solution. Specifically, each element ${\left[ {{{{\bf{\bar v}}}_c}} \right]_n}$  is optimized sequentially while keeping all other elements fixed. To this end, we have the following problem:
\begin{align}\label{element_wise}
\mathop {\max }\limits_{{{\left[ {{{{\bf{\bar v}}}_c}} \right]}_n}} \Upsilon \left( {{{\left[ {{{{\bf{\bar v}}}_c}} \right]}_n}} \right)~~~{\rm{s.t.}}~\eqref{C6-b},
\end{align}
where $\Upsilon \left( x \right)$ denotes the value of the objective function in \eqref{common_pinching_beamforming} when ${{{\left[ {{{{\bf{\bar v}}}_c}} \right]}_n}}$ is set to $x$ while other elements are fixed. Notice that problem \eqref{element_wise} contains a single variable and thus it can be solved optimally via a one-dimensional search. The procedure is applied to all PAs iteratively until convergence.

The optimized solution of ${{{{\bf{\bar v}}}_c}}$ is denoted by ${\bf{\bar v}}_c^\star$. By substituting ${\bf{\bar v}}_c^\star$ into \eqref{C6-a}, the optimization problem with respect to $\left\{ {{t_m}} \right\}$ is a convex problem, which can be solved optimally via standard convex optimization tools, such as CVX \cite{yu2006dual}.

\begin{figure}[t!]
\centering
\includegraphics[width=2.4in]{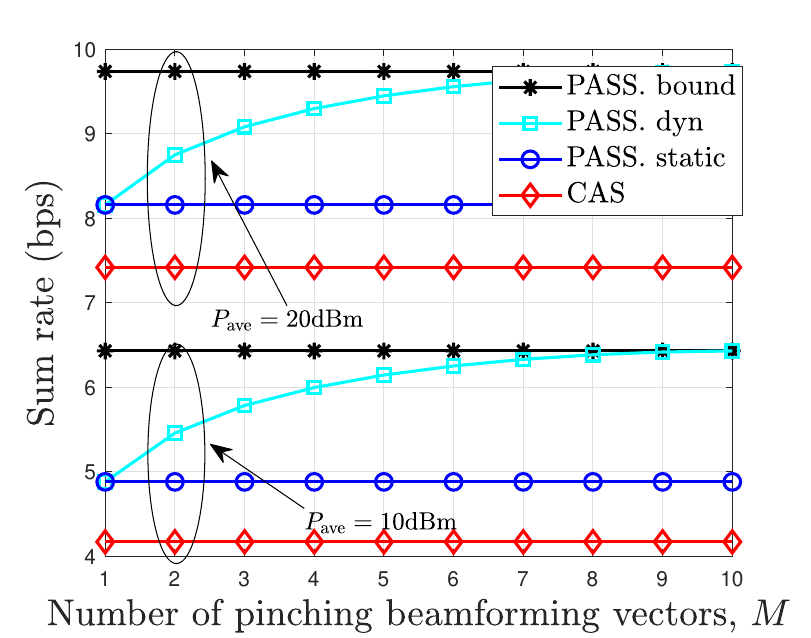}
\caption{{Sum rate versus the number of pinching beamforming vectors}}
\label{sum rate1}
\vspace{-8pt}
\end{figure}

\vspace{-8pt}
\section{Numerical Results}
\vspace{-3pt}
The simulation parameters are set as follows: the noise power is $-80$ dBm, the carrier frequency is 28 GHz, the height of the waveguide is $d = 3$ meter (m), the effective refractive index is ${n_{{\rm{eff}}}} = 1.4$, and the minimum inter-antenna space is $\Delta  = \lambda /2$. There are $K = 10$ users uniformly distributed in a rectangular region centered at the origin with the size of $100 \times 40$ ${{\rm{m}}^2}$. We compare the following schemes: 1) \textbf{CAS}: conventional antenna system, where the antennas are fixed at $\left( {0,0,d} \right)$ m; 2) \textbf{PASS. static}: A static beamforming is shared by all users for uplink transmission; 3) \textbf{PASS. dyn}: $M$ pinching beamforming vectors are employed to assist the uplink transmission; and 4) \textbf{PASS. bound}: the sum rate achieved in \textbf{Theorem 1}, which serves as the sum rate limit in a PASS assisted MAC.

Under the single-antenna setup, we plot the sum rate achieved by all the considered schemes versus the number of available pinching beamforming vectors in Fig. \ref{sum rate1}. Compared to \textbf{CAS}, it is observed that the achievable sum rate is significantly improved with the aid of PASS, which underscores the capability of PASS to reduce the path loss via the flexible position optimization. Moreover, compared to PASS with a static pinching beamforming, substantial sum rate gains can be achieved by employing dynamic pinching beamforming and the resulting performance gains become more pronounced as $M$ increases. The result highlights the potential of dynamically adjust PA positions to create favorable time-selective channels for the uplink transmission. It can be also observed that employing a total number of $M = 7$ pinching beamforming vectors is sufficient to approach the sum rate limit of PASS-assisted  MAC and further increases $M$ only introduces marginal performance gain. Note that employing more pinching beamfomring vectors incurs extra controlling overhead. The result unveils an interesting performance-cost tradeoff in exploiting dynamic pinching beamforming.

\begin{figure}[t!]
\centering
\includegraphics[width=2.4in]{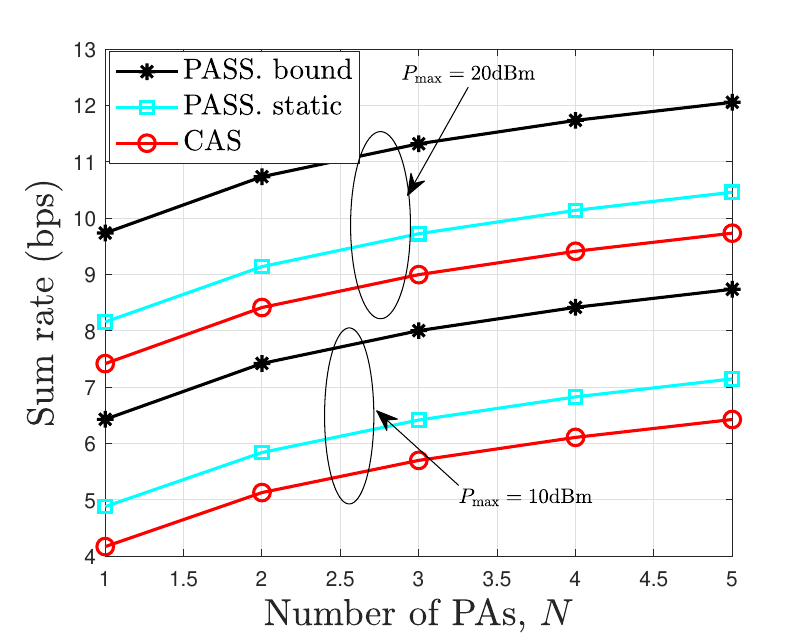}
\caption{{Sum rate versus the number of PAs}}
\label{sum rate2}
\vspace{-8pt}
\end{figure}

In Fig. \ref{sum rate2}, we plot the sum rate versus the number of PAs to evaluate the sum rate limit under the multi-antenna setup. It is observed that the sum rates achieved by all the schemes increase as $N$ increases. The result is expected since using more antennas provides additional degree of freedoms to reconfigure channels. Besides, the sum rate of the capacity-achieving scheme of PASS significantly outperforms other benchmark schemes, which further confirms employing dynamic pinching beamforming can leverage more favorable channels for efficient uplink transmission.

\vspace{-8pt}
\section{Conclusion}
This  paper studied the fundamental limit of the sum rate in PASS-assisted MAC. Theoretical analysis unveiled that the optimal transmission scheme is to carry out alternating transmission among each user whose channel power gain is maximized with the tailored pinching beamforming. The result sheds light on the sum-capacity optimality of OMA in PASS-assisted  MAC. Then, a lower bound of the sum rate under a finite number of pinching beamforming vectors was constructed. Numerical results validated our theoretical findings and also demonstrated the fundamental performance-cost tradeoff in exploiting dynamic pinching beamforming.


\bibliographystyle{IEEEtran}
\vspace{-8pt}
\bibliography{IEEEabrv,myref}


\end{document}